 \documentclass[10pt,twocolumn,twoside]{IEEEtran}

\usepackage{blindtext} 
\usepackage{amsmath,amsfonts,graphicx,amssymb}
\usepackage{bm}
\usepackage{float}
\usepackage{amsthm}
\usepackage{mathtools}
\newtheorem{theorem}{Theorem}[section]
\newtheorem{lemma}{Lemma}[section]

\usepackage{color}
\title{Opportunistic Scheduling as Restless Bandits}

\author{Vivek S.\ Borkar, Gaurav S.\ Kasbekar, Sarath Pattathil,\\ Priyesh Y.\ Shetty$^*$\thanks{*Authors arranged in alphabetical order. VSB, GSK, SP are and PYS was with the Department of Electrical Engineering, IIT Bombay, Powai, Mumbai  400076, India. PYS is now with the Department of Electrical and Computer Engineering, University of California, Davis. Email: borkar.vs@gmail.com, gskasbekar@ee.iitb.ac.in, sarathpattathil@iitb.ac.in, pshetty@ucdavis.edu. Work of VSB was supported in part by a J.\ C.\ Bose Fellowship,  CEFIPRA grant No. IFC/DST-Inria-2016-01/448 ``Machine Learning for Network Analytics''  and a grant for `\textit{Approximation for high dimensional optimization and control problems}' from the Department of Science and Technology, Government of India. }}

\begin{document}




\maketitle
\begin{abstract}
In this paper we consider energy efficient scheduling in a multiuser setting where each user has a finite sized queue and there is a cost associated with holding  packets (jobs) in each queue (modeling the delay constraints).  The packets of each user need to be sent over a common channel. The channel qualities seen by the users are time-varying and differ across users. Also, the cost incurred, \emph{i.e.} energy consumed, in packet transmission is a function of the channel quality. We pose the problem as an average cost Markov Decision Problem and prove that this problem is Whittle Indexable. Based on this result we propose an algorithm in which the Whittle index of each user is computed and the user who has the lowest value is selected for transmission. We evaluate the performance of this algorithm via simulations and show that it achieves a lower average cost than the Maximum Weight Scheduling and Weighted Fair Scheduling strategies.
\end{abstract}

\section{Introduction}
\label{section_introduction}

Recently, there has been a tremendous growth in the deployment of wireless cellular networks around the world, including those based on the popular Long Term Evolution Advanced (LTE-A)~\cite{RF:ghosh:fundamentals:of:lte} standard. A key objective in the design of cellular networks is to \emph{minimize the data transmission delay}, especially that of real-time traffic such as audio or video calls and video streaming. Another important objective is to \emph{minimize the energy consumption} at mobile users and base stations (BS) in order to reduce the energy cost and adverse impact on the environment~\cite{RF:chen:green:wireless:networks}.

In this paper we study the fundamental problem of opportunistic scheduling in a multiuser setting with the objective of minimizing the delay and energy consumption. In this problem there are multiple users, each with a queue of packets  which need to be sent over a common channel. For example, the queues may correspond to different mobile users in a cell wanting to transmit to or receive from the BS over the uplink or downlink wireless channel respectively. The channel qualities seen by the users are time-varying, \emph{e.g.} due to multipath fading of the wireless channel, and differ across users. The energy consumed in packet transmission is a function of the channel quality. At any time at most one user may transmit on the channel because if multiple users were to transmit, there would be  interference. The problem is to select the user (queue) that transmits and to decide the number of packets that the selected queue transmits in each time slot so as to minimize the time-averaged cost, where the cost per slot is an increasing function of the energy consumed in packet transmission and of the delay incurred.

In the model in this paper, the energy required to transmit packets reliably over the channel is an increasing convex function of the rate of transmission, as is typically the case in practice~\cite{Tse}. The packets that are not transmitted by the scheduled user in a given time slot are retained in its queue, which causes delay. These delays can be reduced by transmitting a larger number of packets by using more power. Therefore there is a trade-off between the delay incurred in packet transmission and the energy consumed by transmitters. Note that the delay experienced by a packet is an increasing function of the number of packets ahead of it in its queue. Since our objective is to minimize packet delays, we include a term proportional to the queue length in the objective function, referred to as the ``\emph{holding cost}''. We formulate the problem as an average cost Markov Decision Process (MDP) and prove that it is Whittle indexable~\cite{Whittle}. We use this fact to decouple the problem into individual control problems for each user and propose an algorithm by which the Whittle index of each user is computed and the user who has the lowest value is selected for transmission. We evaluate the performance of this algorithm via simulations and show that it achieves a lower average cost than the Maximum Weight Scheduling and Weighted Fair Scheduling strategies.

We now briefly review related prior literature. A survey of  techniques for energy efficient scheduling with delay constraints in a wireless setting can be found in~\cite{Modiano}. The problem of energy efficient scheduling under delay constraints was first introduced in~\cite{Randall}. This paper studies the tradeoff between minimizing delay and minimizing transmit power for transmission over a block fading wireless channel. The problem is solved by a Markov decision formulation for which a Pareto optimal solution is obtained. The problem of scheduling under power constraints for a fixed deadline is formulated and an offline algorithm to solve it  is proposed in~\cite{Prabhakar}. In~\cite{Bacinoglu},  a similar problem over a finite horizon is formulated and an online heuristic algorithm to solve it is proposed. There are numerous other works (for example see~\cite{Agarwal} and the references therein) that generalize the arrival processes and channel states, and characterize the optimal power delay tradeoff curves. However, all these works deal with the single user case in which there is only one transmitter, whereas we study the multiuser case in this article.

Energy efficient scheduling with delay constraints in a multiuser setting has been explored in~\cite{Salodkar}. In the scheme proposed therein, each user solves a single user power-minimizing delay constrained scheduling problem and finds an optimal rate, which it communicates to the BS. The BS selects the user with the highest rate for transmission. The stability and optimality (in a suitable sense) of this algorithm have also been studied. In~\cite{Multiclass_Queue}, multiuser scheduling with a single server is considered when there are costs associated with holding jobs in each queue and there is a corresponding reward associated with transmission. The costs are similar to the holding costs in queues, which characterize delay requirements in our paper. The problem is formulated as an infinite horizon MDP and the difference of the net reward and the holding cost is maximized. In~\cite{Moghaddari},~\cite{Moghadari},  delay minimization under power constraints for uplink transmission in a multiuser wireless setting is studied. The problem is modeled as an average cost MDP, and an online stochastic approximation algorithm is proposed which is distributed, has low complexity, and converges to the optimal solution to the problem. In~\cite{Zhang}, the question of how the transmit power needs to increase as the delay requirement becomes stringent is studied. Also, the problem of minimizing the transmit power subject to a delay constraint that is in terms of the queue length decay rate is addressed for both  the single user as well as the  multiuser case. However, none of the above papers~\cite{Salodkar},~\cite{Multiclass_Queue},~\cite{Moghaddari},~\cite{Moghadari},~\cite{Zhang} show Whittle indexability of the respective opportunistic scheduling problems they address. To the best of our knowledge, this paper is the first to show Whittle indexability of the opportunistic scheduling problem in a multiuser setting with the objective of minimization of delay and energy consumption. The fact that this problem is Whittle indexable allows us to decouple the original multiuser average cost MDP which is difficult to solve directly, into more tractable individual control problems for each user. In particular, if each queue has an identical buffer size, then it is easy to see that the size of the state space grows exponentially in the number of queues for the original problem and linearly for the decoupled problems. For a precise hardness result for restless bandits, see \cite{Papa}. The decoupling leads to an efficient algorithm  for computation of Whittle indices. As we shall see, the Whittle index policy is empirically found to  outperform widely used heuristics such as the Maximum Weight Scheduling and Weighted Fair Scheduling strategies.

It should be kept in mind, however, that Whittle index policy is itself a heuristic, as the aforementioned decoupling is achieved by first relaxing the original problem to a more analytically amenable one (see Section~\ref{section_model} below). It is known to be optimal in an asymptotic sense in the `infinitely many bandits' limit \cite{Weber}. More importantly, it has been found to be very successful in many applications, see, e.g., \cite{ephemeral, Cloud, proc_sharing, Cowan, Gittins, Jacko, Larranaga, Liu, Nino, Ny, Raghu, Ruiz}. It is also worth noting that the specific problem considered here has a novel feature of being a combination of a restless\footnote{Note that in the problem addressed in this paper, the queue lengths of the queues that do not transmit in a given slot may change due to the arrival of packets; hence, this problem is an instance of the ``restless'' bandit problem~\cite{Whittle}.}  bandit (optimization over choice of bandits) and a conventional MDP  (optimization over number of packets to be transmitted).

The rest of this paper is organized as follows. In Section~\ref{section_model}, we describe the model and problem formulation and provide a review of the theory of Whittle index. In Section~\ref{section_DP}, we show that the optimization problem formulated in Section~\ref{section_model} gets decoupled into individual control problems for each queue and derive a dynamic programming equation for each queue. In Section~\ref{section_struc_prop_vfunc}, we show some  important structural properties of the value function and in Section~\ref{SSC:optimality:threshold:policy}, we show that the optimal policy for the individual control problems is a threshold policy. The properties proved in Section~\ref{SC:dynamic:programming:optimal:policy}   are then used in Section~\ref{section_whittle_index} to prove Whittle indexability of the above problem. In Sections~\ref{section_other_policies} and~\ref{section_simulations}, we present some other scheduling policies for the opportunistic scheduling problem and compare the proposed Whittle index based scheme with these policies via simulations. Finally, we conclude in Section~\ref{section_conclusions}.

\section{Model, Problem Formulation and Background}
\label{section_model}


There are a total of $L$ users, each with a queue of packets, wanting to transmit on a channel (see Figure~\ref{fig:Network_Model}). Time is divided into slots of equal duration. In any time slot, at most one user may transmit on the channel since if multiple users were to transmit simultaneously, their transmissions would interfere with each other. We study the scheduling problem of selecting the user (queue) that is \emph{active}, \emph{i.e.}, transmits, and deciding the number of packets that it transmits, in each time slot.   We consider Poisson arrivals into the queues, where arrivals into queue $i$ are i.i.d.\ Poisson with parameter $\Lambda_i$. When a queue is active, packets may arrive to and / or depart from it, whereas when a queue is \emph{passive}, \emph{i.e.}, does not transmit, packets may arrive to, but not depart from it.

The $i$th queue has a buffer size $M^i$. So, if this queue has $M^i$ packets, all arrivals to it until a packet from it departs are discarded.  Thus, the number of packets in the queue at any time is in the range $\{0, 1, \cdots, M^i\}$.

The per job (packet) holding cost in queue $i$ is $C^i$. By this, we mean that if there are $k$ jobs in queue $i$, the cost incurred in holding these jobs is $k C^i$. This cost models the delay requirement for a queue; in particular, more stringent the delay requirement\footnote{For example, the delay requirement of queues that store delay-sensitive traffic (\emph{e.g.}, voice, video) would be more stringent than those that store elastic traffic (\emph{e.g.}, file transfer).} of user $i$, higher would be the value of $C^i$.

We assume that the channel quality seen by each user is an irreducible finite Markov chain taking values in a discrete set of real numbers (which is tantamount to quantizing the possible values  thereof) and that the channel qualities of different users are independent\footnote{ The assumption that fading is independent across users would be a good approximation for a scenario where different users are situated at mutually far apart locations (e.g., the users may be mobiles in a macrocell); in this case, different users would experience different levels of path loss, shadow fading and multipath fading.}. The next channel state as seen by queue $i$ is governed by the transition kernel $q^i(\mu^i_n, dw)$, where $\mu^i_n$ is the current state of the channel for queue $i$ in time slot $n$. The states of the channel are such that, larger the value of the state, the more noisy the channel and therefore, the more the amount of power that is required for packet transmission. We assume that $q^i(a, dw)$ is First Order Stochastically Dominant (FSD) over $q^i(b, dw)$, when $a>b$. What this essentially means is that if the channel is in a noisy state in one time slot, the probability of being in a bad state in the next time slot is higher as compared to the probability of a good channel state moving to a bad one.

Let $X^1_n, X^2_n, \cdots, X^L_n$ denote the number of jobs that are present in time slot $n$ in queues  $1,2,\cdots ,L$ respectively.

\begin{figure}[h]
\begin{center}
\includegraphics[angle=0,scale=0.3]{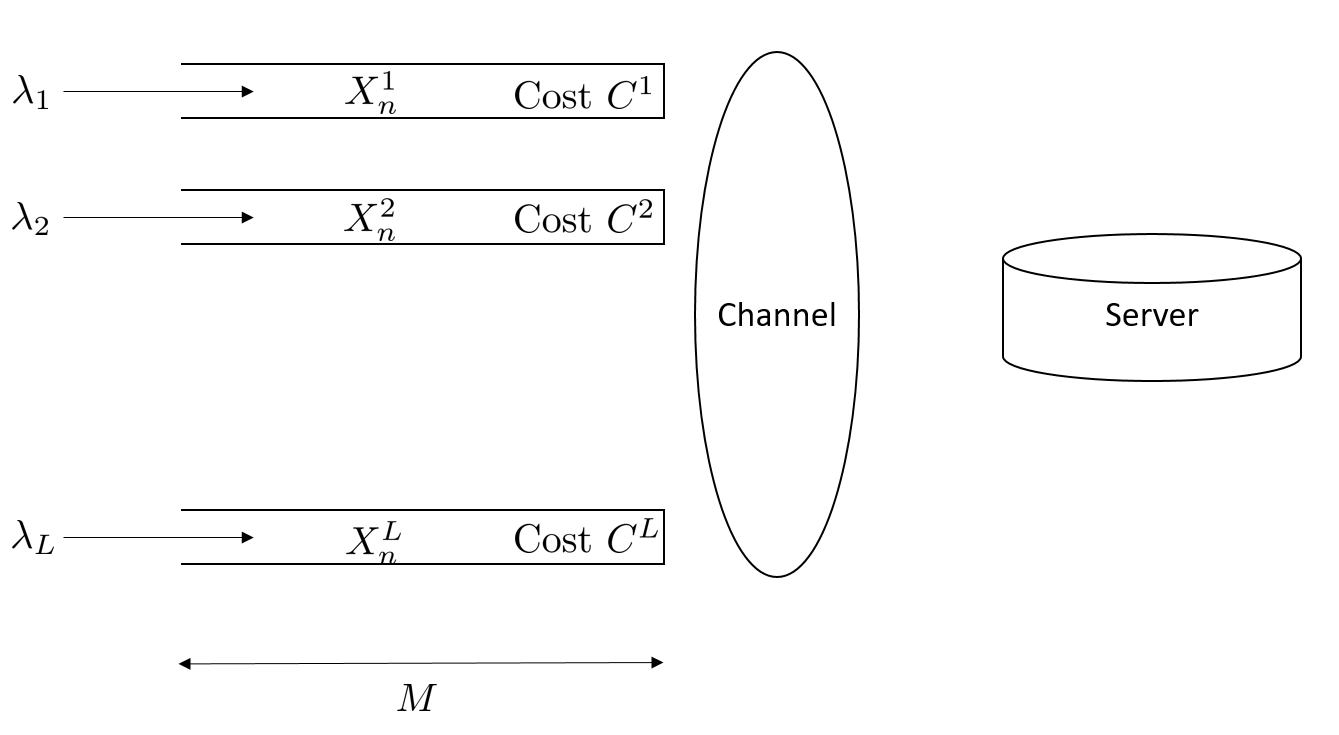}
\caption{The network model described in  Section~\ref{section_model}.}
\label{fig:Network_Model}
\end{center}
\end{figure}

\noindent The dynamics of queue $i$ are given by:
\begin{align}
X^i_{n+1} = (X^i_n - U_n^i Z^i_n + K^i_{n+1}) \wedge M^i
\label{Dynamics}
\end{align}

\noindent where  $K^i_{n+1}$ is the number of arrivals into queue $i$ in time slot $n + 1$ and  $U_n^i$ is a $\{0, 1\}$-valued control variable for queue $i$ with the interpretation that in time slot $n$, $U^i_n = 1 \Leftrightarrow$ the queue is active and $U_n^i = 0 \Leftrightarrow$  the queue is passive. $Z^i_n \in \{0, \cdots, X_n^i\}$ is the number of packets transmitted from queue $i$ in time slot $n$. Also, $a \wedge b$ denotes the minimum of $a$ and $b$. Since only one queue may transmit in any time slot, we have the following constraint:
\begin{align}
\sum_{i = 1}^{L} U_n^i \leq 1 \qquad \forall n
\end{align}
Let $f^i(z)$ be the ``energy cost'' associated with queue $i$ for transmitting $z$ packets; in particular, the cost of transmitting $z$ packets from queue $i$ when the channel state is $\mu$ is $\mu f^i(z)$. We assume that $f^i(\cdot)$ is a convex increasing function and $f^i(0) = 0$.

Our objective is to minimize the time-averaged cost; hence, the problem we address can be stated as:
\begin{align}
\min \lim_{N \uparrow \infty} &\frac{1}{N} \sum_{n=0}^{N-1} \sum_{j = 1}^{L} \mathbb{E}[U_n^j \mu_n^j f^j (Z_n^j) + C^jX_n^j] \label{EQ:objective:function} \\
&\text{s.t.}  \sum_{i = 1}^{L} U_n^i \leq 1 \qquad \forall n. \label{EQ:sum:Uni:eq:1}
\end{align}
The hard per stage constraint (\ref{EQ:sum:Uni:eq:1}) makes the problem hard \cite{Papa}. For this reason, Whittle introduced a relaxation of the per stage constraint (\ref{EQ:sum:Uni:eq:1}) by a time-averaged constraint
\begin{align}
\lim_{N \uparrow \infty} \frac{1}{N} \sum_{n=0}^{N-1} \sum_{j = 1}^{L} \mathbb{E}[U_n^j] \leq 1,
\end{align}
which is a significant relaxation of the former. In particular, an optimal strategy for the latter need not even be feasible for the former. The advantage of this drastic step is that now the constraint is of the same form, viz., time-averaged, as the cost (\ref{EQ:objective:function}). This makes it a classical `constrained MDP' \cite{Borkar}. This can be cast as an abstract linear program in terms of the so called `ergodic occupation measures' which facilitates the application of convex analysis techniques (\textit{ibid.}). Of relevance to us here  is the fact that classical  Lagrange multiplier formulation is now possible and leads to the following unconstrained problem:
\begin{align}
\min \lim_{N \uparrow \infty} \frac{1}{N} \sum_{n=0}^{N-1} \sum_{j = 1}^{L} \mathbb{E}[U_n^j \mu_n^j f^j (Z_n^j) + C^jX_n^j + (1-U_n^j)\lambda].
\end{align}
Here $\lambda$ is the Lagrange multiplier. Whittle's master stroke was to take away the identity of $\lambda$ as the Lagrange multiplier and view it as a negative subsidy or `tax' for passivity\footnote{negative because this is a cost minimization problem. The original Whittle formulation is for a reward maximization problem, we give here equivalent statements for a cost minimization problem. Also, we have replaced his equality constraint by an inequality constraint. The overall philosophy, however, is identical.}. The relaxed  problem has a separable cost and a separable constraint. Hence given $\lambda$, it decouples into individual control problems
\begin{align}
\min \lim_{N \uparrow \infty} &\frac{1}{N} \sum_{n=0}^{N-1} \mathbb{E}[U_n^j \mu_n^j f^j (Z_n^j) + C^jX_n^j] \label{EQ:objective:function2} \\
&\text{s.t.}  \ \  \lim_{N \uparrow \infty} \frac{1}{N} \sum_{n=0}^{N-1} \mathbb{E}[U_n^j] \leq 1 \label{EQ:sum:Uni:eq:2}
\end{align}
for each $j$. Whittle then defines \textit{indexability} (now called \textit{Whittle indexability}) as the property: the set of states that are passive under optimal policy decreases monotonically  from the whole state space to the empty set as $\lambda$ is increased monotonically from $-\infty$ to $+\infty$. \textit{If} the problem is Whittle indexable, then the  (Whittle) index is defined for each $j$ and state $(x, \mu)$ as the value $\lambda^j(x, \mu)$ of $\lambda$ for which both active and passive modes are equally desirable for the $j$th control problem (\ref{EQ:objective:function2})-(\ref{EQ:sum:Uni:eq:2}). (If this choice is not unique, we take the least such $\lambda$ in order to render it unambiguous. This will be implicitly assumed throughout what follows.) The control policy then is as follows:  in time slot $n$, arrange $\{\lambda^j(X_n^j, \mu^j_n)\}$ in decreasing order (any tie being resolved according to some fixed tie-breaking rule) and  then select the $j_n$'th queue for transmission, where  $j_n :=$ argmin$_j\lambda^j(X_n^j, \mu^j_n)$ if $\min_j \lambda^j(X_n^j, \mu^j_n) < 0$. However, if $\min_j \lambda^j(X_n^j) > 0$, we choose not to allow any queue to transmit.

If one were to treat this as a classical average cost constrained MDP, one can indeed decouple the problem into individual unconstrained control problems of minimizing
\begin{align}
\lim_{N \uparrow \infty} &\frac{1}{N} \sum_{n=0}^{N-1} \mathbb{E}[U_n^j \mu_n^j f^j (Z_n^j) + C^jX_n^j + \lambda^*(1 - U_n^j)]
\end{align}
where $\lambda^*$ is the Lagrange multiplier which needs a separate computation \cite{Borkar}. If one solves this problem, the possibility of more than one chain being active cannot be eliminated, because only on average the number of active bandits will be one. This situation is infeasible for the original problem.

\section{Dynamic Programming and Optimal Policy}
\label{SC:dynamic:programming:optimal:policy}
\subsection{The Dynamic Programming Equation}
\label{section_DP}

Given the value of $\lambda$,
 the optimization problem gets decoupled into individual control problems for each one of the queues separately.
Since the problem gets decoupled, we henceforth drop the superscript in each of the variables.
Each individual problem above is a classical average cost MDP.
The dynamic programming equation for each queue can be derived by a vanishing discount argument as in \cite{Agarwal} and is:
\begin{align}
V(x,\mu) = -\beta &+ Cx + \min [ \min_z (\mu f(z) +  \nonumber \\
&\int  \int V(y,w)  p_1(dy|z,x) q(\mu, dw) ), \nonumber \\
&\lambda + \int \int V(y,w)p_0 (dy|x)q(\mu, dw) )]. \label{EQ:DP}
\end{align}
Here,
 \begin{itemize}

 \item $\beta$ is the optimal value of the average cost problem,

 \item $p_1(\cdot|z,x)$ is the transition probability when the queue is active, there are $x$ jobs in the queue and $z$ jobs are being transmitted,

 \item $p_0(\cdot|x)$ is the transition probability when the queue is passive,  there are $x$ jobs in the queue and there are no transmissions.

 \end{itemize}

 Note that the event `all buffers become full at time $n$' has a non-zero probability. Thus this Markov chain has a `uni-chain' property, whence (\ref{EQ:DP}) uniquely specifies $\beta$ as the optimal cost and uniquely specifies $V$ up to an additive constant \cite{Puterman}. We render $V$ unique by adding the requirement $V(x_0, \mu_0) = \beta$ for a prescribed $(x_0, \mu_0)$.

In the following subsections, we prove some important structural properties of the value function $V(\cdot)$ in (\ref{EQ:DP}) and show that the optimal policy for the individual control problems is a threshold policy in the state variable with a threshold that depends on the channel state. That is, there is a function of the channel state taking values in the state space of the queue such that, if the current state of the queue is greater than or equal to the value of this function, then the queue is active, and passive if not.  This is used in Section~\ref{section_whittle_index} to prove Whittle indexability of the above problem. We closely follow the approach of \cite{Agarwal}, but include most key details in toto for sake of making this account reasonably self-contained.

\subsection{Monotonicity and Convexity of the Value Function}
\label{section_struc_prop_vfunc}

The key property we need is the convexity of $V$, proved below.

\begin{lemma}
\label{inc_arg}
$V(\cdot,\mu)$ is an increasing function for every fixed $\mu$.
\end{lemma}
\begin{proof}
Let $f_\mu(.) = \mu f(.)$. Fix $\lambda$, the control processes $\{U_n\},\{Z_n\},$ and arrival process $\{K_n\}$ on a probability space and consider two state processes $\{X'_n\},\{X_n\}$ driven by these according to (\ref{Dynamics}) with initial conditions $X'_0 = x' > x =  X_0$. Then $X'_n > X_n \: \: \forall n$ and therefore,
\begin{align}
CX'_n + f_\mu(U_n Z_n) + (1 - U_n)\lambda > \nonumber \\
CX_n + f_\mu(U_n Z_n) + (1 - U_n)\lambda \: \: \forall n.
\end{align}
Let
\begin{align}
J_x^\alpha (\{U_n\},\{Z_n\}) := \mathbb{E}[\sum_{m=0}^{\infty}\alpha^m(CX_m \nonumber \\
+ f_\mu(U_mZ_m) + (1 - U_m)\lambda)]
\end{align}
denote the $\alpha$-discounted cost for initial condition  $x$ with the given control processes. (Here the expectation is taken on arrivals as well as the channel states.) Then
$$ J_{x'}^\alpha (\{U_n\},\{Z_n\})\geq J_x^\alpha (\{U_n\},\{Z_n\}).$$
Taking minimum over all control processes on both sides, the discounted value functions $V_\alpha(\cdot,\mu)$ satisfy $V_\alpha(x',\mu) \geq V_\alpha(x,\mu)$. Using the vanishing discount argument (see \cite{Agarwal}), the claim extends to average cost value function $V(\cdot,\mu)$.
\end{proof}

\begin{lemma}
\label{inc_diff_channel_state}
$V(x,\cdot)$ is increasing in the channel state for every fixed $x$.
\end{lemma}
\begin{proof}
The proof goes along the same lines as the previous lemma, along with the fact that the channel state transition probabilities satisfy the stochastic dominance condition. See Proof of Theorem 2 in \cite{Agarwal} for more details.
\end{proof}

This result indicates that one can prove structural properties in the channel state variable $\mu$ analogous to those for the queue state variable $x$. We do not do so because while the two jointly form the overall state of the dynamics under consideration, the channel state is uncontrolled. Further, as pointed out at the beginning of  section III.b, p.\ 1480, of \cite{Agarwal}, channel state under Markov fading is not conducive to the kind of structural results we obtained for queue state for solid technical reasons.  Hence we treat the channel state $\mu$ as a parameter and prove the structural properties of the value function in $x$ alone holding $\mu$ fixed. This leads to a Whittle index as a function of the queue state with additional dependence on the channel state treated as an extraneous parameter.

\begin{lemma}
\label{inc_diff}
$V(\cdot,\mu)$ is convex and has increasing differences for a fixed $\mu$, \emph{i.e.}, for $z>0, x>y,$ $$V(x+z,\mu) - V(x,\mu) \geq V(y+z,\mu) - V(y,\mu).$$
\end{lemma}
\begin{proof}
Let $f_\mu(.) = \mu f(.)$. For the purposes of this proof, we shall embed the state space in $[0,\infty)$, \emph{i.e.}, treat the non-negative integer valued process as an instance of a non-negative real valued process. However the departure process $\{ Z_n \}$ continues to be an integer valued process constrained to remain  in $[0, \lceil X_n\rceil]$ at time $n$. (The latter stipulation allows the state to go negative at times. This is an artifice of the relaxation to continuous state space which disappears once we restrict to the discrete state space.) The above dynamics makes sense for this scenario as well. We first establish convexity by induction for the finite horizon discounted problems, with discount factor $\alpha$. It is true for horizon $n=0$. Suppose it is true for horizon $n-1$. Let $u_1,z_1$ (resp., $u_2,z_2$) be the optimal decisions for $x_1$ (resp., $x_2$) for the $n$ horizon problem. Without loss of generality, $u_iz_i \leq x_i, i = 1,2$. Then
\begin{align}
V^n(x_i,\mu) = Cx_i + f_\mu(u_iz_i) + (1-u_i)\lambda \nonumber \\
+ \ \alpha\int \int_k V^{n-1}(x_i - u_iz_i + k,w) \xi(dk) q(\mu, dw), \nonumber \\
\qquad \qquad \qquad \qquad \qquad \qquad \qquad    i = 1,2
\end{align}
where $\xi(\cdot)$ is the distribution of arrivals into the system.
Hence
\begin{align}
\lefteqn{\frac{V^n(x_1,\mu) + V^n(x_2,\mu)}{2}} \nonumber \\
 &= C \bigg(\frac{x_1+x_2}{2} \bigg) + \frac{f_\mu(u_1z_1)+ f_{\mu}(u_2z_2)}{2} \nonumber \\
&+ \lambda \bigg( 1 - \frac{u_1 + u_2}{2}\bigg) \nonumber \\
&+ \alpha \int \int_k  \frac{1}{2}\bigg(V^{n-1}(x_1 - u_1z_1 + k, w) \nonumber \\
& \qquad \qquad \qquad + V^{n-1}(x_2 - u_2z_2 + k, w) \bigg) \xi (dk) q(\mu, dw) \nonumber \\
&\geq C \bigg(\frac{x_1+x_2}{2} \bigg) + \frac{f_\mu(u_1z_1)+f_\mu(u_2z_2)}{2} \nonumber \\
&+ \lambda \bigg( 1 - \frac{u_1 + u_2}{2}\bigg) \nonumber \\
&+ \alpha\int \int_k V^{n-1} \bigg( \frac{x_1 + x_2}{2} - \left\lceil\frac{u_1z_1 + u_2z_2}{2}\right\rceil + k, w \bigg) \nonumber \\
& \qquad \qquad \qquad \qquad \qquad \qquad \qquad \qquad \qquad \xi (dk) q(\mu, dw)  \nonumber \\
&\geq V^n \bigg( \frac{x_1 + x_2}{2}, \mu \bigg) \nonumber
\end{align}
by convexity of the functions $f, V^{n-1}$, Lemma~\ref{inc_arg}, and using the fact that
$$ 0 \leq \frac{u_1z_1 + u_2z_2}{2} \leq \frac{x_1 + x_2}{2}.$$
This proves convexity of the finite horizon problem. Convexity is preserved under pointwise convergence, so it follows for the  infinite horizon discounted problem by letting the time horizon go to infinity, and then for the average cost problem by the `vanishing discount' argument as in \cite{Agarwal}. Convexity implies increasing differences. Therefore $V(\cdot,\mu)$ has increasing differences. The function restricted to the non-negative integers will retain this property, thereby proving the lemma.
\end{proof}

\subsection{Optimality of Threshold Policy}
\label{SSC:optimality:threshold:policy}

A threshold policy is one where there is some threshold $x^*$ such that whenever the state of the system $x>x^*$, the optimal decision would be to go active (or passive) and if $x<x^*$, the optimal decision would be to go passive (or active).
The preceding lemma has the following important consequence.

\begin{lemma}
\label{inc_trans}
The map $$x \mapsto  \underset{z}{\operatorname{argmin}} (\mu f(z) + \int \int V(y,w) p_1(dy|z,x) q(\mu, dw))$$ is increasing for fixed $\mu$.
\end{lemma}
\begin{proof}
Let $z' \geq z$, $x' \geq x$ and $z',z \leq x$. From the increasing differences property (Lemma~\ref{inc_diff}) we have that $\forall \mu$:
\begin{align}
&V(x' - z + k,\mu) - V(x' - z' + k,\mu) \nonumber \\
& \qquad \qquad \qquad \geq  V(x - z + k,\mu) - V(x - z' + k,\mu).
\end{align}
This gives us:
\begin{align}
\int \int_k [ V(x' - z + k,w) - V(x' - z' + k,w) ] \xi(dk) q(\mu, dw) \nonumber \\
\geq \int \int_k [V(x - z + k,w) - V(x - z' + k,w)]\xi(dk) q(\mu, dw).
\label{eq:diff}
\end{align}
Define:
$$h_\mu(z,x)= \mu f(z) + \int \int_k V(x - z + k,w) \xi(dk) q(\mu, dw). $$
Using this definition of $h_\mu(z,x)$ and equation (\ref{eq:diff}), we have
\begin{align}
h_\mu(z',x') - h_\mu(z,x') \leq h_\mu(z',x) - h_\mu(z,x).
\label{dec:diff}
\end{align}
This shows that $h_\mu(z,x)$ is a submodular function or in other words $-h_\mu(z,x)$ is a supermodular function. We also have:
\begin{align*}
& \underset{z}{\operatorname{argmin}} (\mu f(z) + \int \int V(y,w) p_1(dy|z,x) q( \mu, dw)) \nonumber \\
& \qquad \qquad \qquad = \underset{z}{\operatorname{argmin}} \: h_\mu(z,x) \\
& \qquad \qquad \qquad = \underset{z}{\operatorname{argmax}} \: -h_\mu(z,x).
\end{align*}
Using Theorem 10.7, Pg 259 \cite{Sundaram}, we get the desired result.
\end{proof}

\begin{lemma}
\label{lemma_threshold}
The optimal policy is a threshold policy.   That is, for each fixed $\mu$, $\exists$ a threshold $x^*$ such that if $x \geq x^*$ (respectively, $x < x^*$), it is optimal to transmit (respectively, not transmit) in state $x$.
\end{lemma}
\begin{proof}
Define
\begin{align*}
g(x,\mu) &= f_\mu(z_1) \ + \\
& \qquad \mathbb{E}[V(x - z_1 + \xi ,w)] - \mathbb{E}[V(x + \xi,w)]
\end{align*}
where $z_1$ is the optimal number of departures for $x$ when the channel state is  $\mu$. The next arrival  is denoted by $\xi$ and the next channel state is denoted by $w$. Here, we assume the channel state $\mu$ is fixed. Expectation is taken over the next channel state and arrival.
We will show that $x \mapsto g(x,\mu)$ is a decreasing function, or equivalently $g(x+1,\mu)- g(x,\mu) \leq 0$. The result will then follow from (\ref{EQ:DP}).

Let $z_2$ be the optimal number of departures for $(x+1)$ (for channel state $\mu$). We have $z_2 \geq z_1$ from Lemma \ref{inc_trans}. Consider the following:
\begin{align}
g(x+1,\mu) &= f_\mu(z_2) + \mathbb{E}[V((x+1) - z_2 + k,w )] \nonumber \\
& \qquad \qquad \qquad - \mathbb{E}[V((x+1) + k,w)] \nonumber \\
&\leq^{*1} f_\mu(z_1) + \mathbb{E}[V((x+1) - z_1 + k,w )] \nonumber \\
& \qquad \qquad \qquad - \mathbb{E}[V((x+1) + k,w)] \nonumber \\
&\leq^{*2} f_\mu(z_1) - \{ \mathbb{E}[V(x + k ,w)] \nonumber \\
& \qquad \qquad \qquad - \mathbb{E}[V(x - z_1 + k,w)] \} \nonumber \\
&= g(x,\mu). \nonumber
\end{align}
Note that $*1$ follows from the definition of $z_2$ and
$*2$ is a direct consequence of Lemma~\ref{inc_diff}.
\end{proof}

For later use, we also prove the following result wherein we write $V$ as $V_{\lambda}$ to render explicit its dependence on $\lambda$.
\begin{lemma}\label{lambda}
The map $\lambda \mapsto V_{\lambda}(x,\mu)$ is concave increasing for fixed $x, \mu$. In particular, it is  continuous.
\end{lemma}

\begin{proof}
For the discounted cost problem with a fixed control process, it is easy to see that the cost is linear increasing in $\lambda$. The value function, being the minimum thereof over all control processes, will be concave increasing. Concavity and monotonicity is preserved in the vanishing discount limit, proving the claim.
\end{proof}

\section{Whittle indexability and Computation of the Whittle Index}
\label{section_whittle_index}

\subsection{Whittle Indexability}
\begin{theorem}
\label{whittle_index}
The above problem is Whittle indexable, where we parametrize the Whittle index by the channel state $\mu$.
\end{theorem}
\begin{proof}
Fix the channel state to be $\mu$. We suppress the $\mu$-dependence of optimal thresholds in what follows for  notational ease. Let $\lambda' > \lambda$ and the corresponding optimal thresholds (which exist by Lemma~\ref{lemma_threshold}) be $x^*(\lambda'), x^*(\lambda)$ respectively. Suppose $x^*(\lambda') > x^*(\lambda)$. We have:
\begin{align}
\lambda &= \min_{z} \big[ \mu f(z) + \int \int_k \big( V(x^*(\lambda) - z + k,w) \nonumber \\
& \quad \qquad - V(x^*(\lambda)+ k,w) \big) \xi(dk) q(\mu ,dw) \big]  \nonumber \\
&= \mu f(z_{\lambda}(x^*(\lambda))) + \int \int_k \big( V(x^*(\lambda) - z_{\lambda}(x^*(\lambda)) + k,w) \nonumber \\
& \qquad \qquad \qquad - V(x^*(\lambda)+ k,w) \big) \xi(dk) q(\mu ,dw) \nonumber
\end{align}
where $z_{\lambda}(x^*(\lambda))$ is the optimal transmission from state $x^*(\lambda)$.
\newline
\newline
Since $\lambda' > \lambda$, we have:
\begin{align}
\lambda' &> \mu f(z_{\lambda}(x^*(\lambda))) + \int \int_k \big( V(x^*(\lambda) - z_{\lambda}(x^*(\lambda)) + k,w)  \nonumber \\
& \qquad \qquad \qquad  - V(x^*(\lambda)+ k,w) \big) \xi(dk) q(\mu, dw) \nonumber \\
&\geq^{*1} \mu f(z_{\lambda}(x^*(\lambda))) \nonumber \\
& \qquad \qquad \qquad  + \int \int_k \big( V(x^*(\lambda')- z_{\lambda}(x^*(\lambda)) + k,w) \nonumber \\
& \qquad \qquad \qquad - V(x^*(\lambda')+ k,w) \big) \xi(dk) q(\mu, dw) \nonumber \\
&\geq  \min_{z} \big[ \mu f(z) + \int \int_k \big( V(x^*(\lambda') - z + k,w) \nonumber \\
& \qquad \qquad \qquad - V(x^*(\lambda')+ k,w) \big) \xi(dk)  q(\mu, dw) \big]  \nonumber \\
&= \lambda'. \nonumber
\end{align}
Here ($*1$) follows from Lemma~\ref{inc_trans}, since $x^*(\lambda) < x^*(\lambda')$. However, this leads to a contradiction.
Therefore $x^*(\lambda)$ is a decreasing function of $\lambda$ for a fixed channel state $\mu$. The set of passive states for $\lambda$ is given by $ [0,x^*(\lambda)]$. Since $x^*(\lambda)$ is a decreasing function of $\lambda$, we have that the set of passive states monotonically decreases to $\phi$ as $\lambda \uparrow \infty$. This shows Whittle indexability.
\end{proof}

\subsection{Computation of the Whittle index}

We sketch now an algorithm for computation of the Whittle index for each threshold $x$ and channel state $\mu$.
Recall that the dynamic programming equation for an individual queue is given by:
\begin{align}
V_\lambda(x,\mu) = &\min_{u \in \{0,1\}, z\in [0,x]} \Big[ Cx + u \mu f(z) + (1-u) \lambda - \beta \nonumber \\
&  + \int \int_k V_\lambda(x - uz + k ,w) \xi(dk) q(\mu ,dw)\Big] \label{DPindividual}
\end{align}
where we have rendered explicit the $\lambda$-dependence of $V$. The Whittle index is computed using the following set of equations:
\begin{align}
V^{n+1}(x,\mu) &= Cx + \min \big[ \min_{0 \leq z \leq x}( \mu f(z) \nonumber \\
&+ \int \int_k  V^{n}(x - z + k ,w) \xi(dk) q(\mu, dw)),  \nonumber \\
&  \lambda_n(x,a) + \int \int_k V^{n}(x + k ,w) \xi(dk) q(\mu, dw) \big] \nonumber \\
& - V^n(x_0,\mu_0), \label{RVI} \\
\lambda_{n+1} (x,\mu) &= \lambda_{n}(x,\mu) + \gamma \big[ \min_{0 \leq z \leq x}( \mu f(z) \nonumber \\
&+ \int \int_k V^{n}(x - z + k ,w) \xi(dk) q(\mu ,dw)) \nonumber \\
& -\lambda_n(x,\mu) - \int \int_k V^{n}(x + k ,w) \xi(dk) q(\mu, dw) \big], \label{SA}
\end{align}
where $x_0, \mu_0$ are fixed choices as before, and $\gamma > 0$ is a small step-size or `learning parameter'.

If $\lambda_n \equiv$ a constant, (\ref{RVI})  is simply the classical relative value iteration for solving average cost dynamic programming equations \cite{Puterman}. The way to analyze the joint scheme (\ref{RVI})-(\ref{SA}) is to view it as a two time scale algorithm (\cite{book}, Chapters 6,9). Thus the iteration (\ref{RVI}) takes place on the `natural' time scale defined by the iteration index $n = 0, 1, 2, \cdots$, whereas iteration (\ref{SA}) is an incremental adaptation scheme which evolves on a much slower time scale $m = 0, \gamma, 2\gamma, \cdots$. The latter can be viewed as a constant stepsize stochastic approximation algorithm. Using the arguments of \cite{book}, pp.\ 113-115, we can view (\ref{SA}) as quasi-static, i.e., $\lambda_n \approx$ a constant in  order to analyze (\ref{RVI}). Then it is a classical relative value iteration scheme which converges to the value function $V$ of (\ref{DPindividual}) corresponding to $V(x_0, \mu_0) = \beta$, which renders it unique. What this translates into is that $V^n$ tracks $V_{\lambda_n}$, i.e.,
$$\|V^n - V_{\lambda_n}\| \approx 0$$
for small $\gamma$ and sufficiently large $n$. This allows us to view (\ref{SA}) itself as an approximate  discretization  (approximate because of the additional error $V^n - V_{\lambda_n}$) of the ordinary differential equation (ODE)
\begin{eqnarray}
\lefteqn{\dot{\lambda}_t (x,\mu) =  - \lambda_t(x, \mu) + \min_{0 \leq z \leq x}[ \mu f(z)  } \nonumber \\
&&+ \ \int \int_k V_{\lambda_t}(x - z + k ,w) \xi(dk) q(\mu ,dw)] \nonumber \\
&& - \ \int \int_k V_{\lambda_t}(x + k ,w) \xi(dk) q(\mu, dw), \label{ode}
\end{eqnarray}
 where $\lambda_t(x,\mu)$ is the $(x,\mu)$-th component of $\lambda_t$. This is  a scalar ODE of the form
$$\dot{\lambda}_t = F(\lambda_t) - \lambda_t$$
where for $\lambda = [[\lambda(x,\mu)]]$, the $(x,\mu)$th component of $F$ is given by
\begin{eqnarray*}
&& min_{0 \leq z \leq x}[ \mu f(z) +  \int \int_k V_{\lambda_t}(x - z + k ,w)\times \\
&&\xi(dk) q(\mu ,dw)] - \int \int_k V_{\lambda_t}(x + k ,w) \xi(dk) q(\mu, dw).
\end{eqnarray*}
By Lemmas \ref{inc_diff} and \ref{lambda}, the function $F$ is continuous monotone decreasing. Thus (\ref{ode}) will have a unique stable equilibrium to which it must converge.
The iterates $\{\lambda_n\}$ then converge to a small neighborhood of this equilibrium by Theorem 1, p.\ 339, of \cite{Hirsch}. The equilibrium is characterized by setting $F(\lambda) = \lambda$, whence it is seen that it is precisely the Whittle index for the pair $(x, \mu)$.

 To calculate the number of packets transmitted by an active user, we use the equation:
\begin{align}
z^*(x,\mu)  &= \underset{z \in [0,x]}{\operatorname{argmin}} \big[ \mu f(z) \nonumber \\
&+ \int \int_k V^*(x - z + k ,w) \xi(dk) q(\mu, dw) \big], \label{optargmin}
\end{align}
where $V^*(x) := V_{\lambda(x,\mu)}(x)$. Recall that this transmission occurs at each time for exactly one process, viz., that with the lowest Whittle index. Just as the choice of active bandit based on the Whittle indices is a heuristic, so is this choice of the number of packets to be transmitted, and needs some justification. Before we do so, observe that the Whittle index policy for bandit selection compares current Whittle indices across the bandits, thereby introducing a dependence among the processes: they are no longer decoupled, although the computation to arrive at the policy treated them as such. For the obvious computational advantages of such `decoupled thinking' to be retained, one must come up with a heuristic for choosing the number of packets transmitted to respect such decoupling. The most naive choice would be to use the optimal choice thereof given by the single agent problem analyzed in \cite{Agarwal}. But unlike the single agent problem, the individual chains do not, or rather, are not allowed to, transmit except when the corresponding Whittle index wins over the others. This leads to serious under-performance. Intuition suggests that when they do transmit, they should transmit more than what the single agent optimal policy suggests. Clearly the Whittle index has to step in, being a handy function of individual states that  couples the processes. This is what the above heuristic does. Let $\beta^*(x) = \beta(\lambda(x,\mu))$. The definition of Whittle index then leads to the following equation:
\begin{align*}
V^*(x) &= \underset{z \in [0, x]}{\operatorname{min}}\Big[Cx + \mu f(x) -  \beta^*(x) + \\
& \int \int V^*(x - z + k, w)\xi(dk)q(\mu, dw)\Big].
\end{align*}
This amounts to an MDP where a state-dependent subsidy $\beta^*(x)$ is offered in a manner that the average optimal cost is zero. Then clearly the optimal number of transmissions will be higher. Thus our heuristic automatically pegs the latter choice at a higher number to compensate for zero transmission in passive states.

\section{Simulations}
In this section, we evaluate the performance of the proposed Whittle index based algorithm and compare it with those of the Max-Weight Scheduling and Weighted Fair Queuing  strategies via simulations. We describe the above two strategies in Section~\ref{section_other_policies} and present the simulation model and results in Section~\ref{section_simulations}.
\subsection{Max-Weight Scheduling and Weighted Fair Queuing Strategies}
\label{section_other_policies}

\subsubsection{Max-Weight Scheduling}
The Max-Weight Scheduling strategy has been extensively used in prior work, \emph{e.g.}, in the context of resource allocation in wireless networks~\cite{RF:georgiadis:resource},~\cite{RF:Tassiulas:stability},~\cite{Tassiulas} and scheduling in input-queued switches~\cite{RF:mckeown:input:queued:switches}. In this strategy, in each time slot $n$, the channel is allocated to the queue with the largest number of packets, \emph{i.e.}, to queue $l_n = \underset{i}{\operatorname{argmax}} \: X_n^i$, where $X_n^i$, $i \in \{1, \ldots, L\}$, is the number of packets in queue $i$ in time slot $n$.

\subsubsection{Weighted Fair Queuing (WFQ)}
The WFQ policy is a  router link-scheduling discipline that is widely used in communication networks~\cite{Virtual_Time}. Informally, under this policy, in any sufficiently long time interval in which queue $i$ is non-empty, it is guaranteed to be selected for transmission in at least a fraction $\frac{w_i}{\sum_{j=1}^L w_j}$ of the time slots, where $w_i$ is the \emph{weight} of queue $i$; see~\cite{Virtual_Time} for a formal description of the WFQ policy. In our simulations, the weight assigned to queue $i$ is its holding cost, \emph{i.e.}, $w_i = C^i$.
\newline
\newline
The number of packets which are transmitted once a queue is selected, for both the Max-Weight policy as well as the Weighted Fair Queuing policy is given by:
\begin{align}
z^*(x,\mu) &= \underset{z \in [0,x]}{\operatorname{argmin}} \big( \mu f(z) \nonumber \\
&+ \int \int_k V^*(x - z + k ,w) \xi(dk) q(\mu, dw) \big)
\end{align}

\subsection{Simulation Model and Results}
\label{section_simulations}
In our simulations, we use the model described in  Section~\ref{section_model}; throughout, we use the values $M = 50$ and $L = 3$. We focus on the case where $f^i(z) = f(z)$, $i \in \{1,2,3\}$; also, we study the cases where $f(z)$ is exponential ($f(z)$ = $2^z - 1$) and quadratic ($f(z)$ = $kz^2$). We assume that $\forall i \in \{1,2,3\}$ and $n = 1, 2, 3, \ldots$, the channel state $\mu_n^i$ can take two possible values: $1$ (good) and $2$ (bad), and that the transition kernel for each channel is the same and is  given by:
\[
q(\cdot| \cdot)=
  \begin{bmatrix}
    0.7 & 0.3 \\
    0.3 & 0.7
  \end{bmatrix}.
\] Also, in our simulations, the average cost (objective function) is given by:
\begin{align}
\frac{1}{T} \sum_{t=0}^{T} \sum_{i = 1}^{L}\mathbb{E}\big[ C^i X^i_t + \delta U^i_t \mu^i_t f(Z^i_t) \big], \label{EQ:simulations:average:cost}
\end{align}
where $\delta$ is a parameter that can be set so as to assign different weights to the holding cost and the transmission cost.

\begin{figure}[h!]
\begin{center}
\includegraphics[angle=0,scale=0.6]{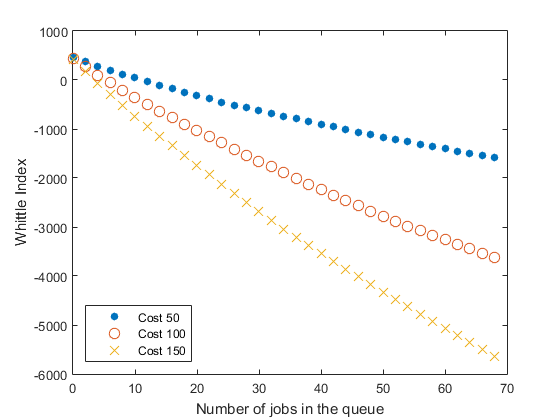}
\caption{Whittle Index for costs 10, 20 and 30 with $f(z) = 2^z - 1$. Channel state $\mu = 1$}
\label{fig:Whittle_Index}
\end{center}
\end{figure}

Let $\Lambda_i = 1$, $i \in \{1,2,3\}$. First, for each of the holding cost values $C = 10, 20$ and $30$, Figure~\ref{fig:Whittle_Index} shows the Whittle index $\lambda(x,\mu)$ versus the queue length $x$.  We see that $\lambda(x,\mu)$ is decreasing in the queue length $x$ for each value of $C$. Also, for each value of $x$, the higher the cost $C$, the lower is the Whittle index value $\lambda(x,\mu)$. The above trends can be interpreted as follows. In the proposed Whittle index based algorithm,
we select the queue $i$ with the lowest value of $\lambda_i(x)$ for transmission. But by the above trends, this results in selection of a queue $i$ with a large queue length $x$ and/ or cost $C^i$, which is consistent with intuition given that our objective is to minimize the cost in (\ref{EQ:simulations:average:cost}).

Next, we compare the performance of the proposed Whittle index based algorithm with those of the Max-Weight Scheduling and WFQ strategies (see Section~\ref{section_other_policies}) in terms of the average cost in \eqref{EQ:simulations:average:cost}.
In Figures~\ref{fig:Cost_Exp_Similar} and~\ref{fig:Cost_Exp_Different}, we have plotted this average cost against the time slot number for the case where the transmission costs are exponential. The holding costs $C^1$, $C^2$ and $C^3$ are $10$, $20$ and $30$ respectively for Figure~\ref{fig:Cost_Exp_Similar} and $10$, $20$ and $500$ respectively for Figure~\ref{fig:Cost_Exp_Different}. It can be seen that in both the figures, the Whittle index based algorithm outperforms the other two strategies. In Figure~\ref{fig:Cost_Exp_Similar}, for which the holding costs ($10$, $20$ and $30$) are close to each other, the Max-Weight Scheduling algorithm performs better than the WFQ algorithm, whereas in Figure~\ref{fig:Cost_Exp_Different}, where there are large differences between the holding costs ($10$, $20$ and $500$), the converse is true. Intuitively, this is because WFQ takes the holding costs into account (through the weight assigned to each queue) and hence prevents the accumulation of a large number of packets (which would result in a high average cost) in the queue with holding cost $500$ resulting in better performance than Max-Weight Scheduling in the scenario of Figure~\ref{fig:Cost_Exp_Different}; on the other hand, in the scenario of Figure~\ref{fig:Cost_Exp_Similar}, the benefit from taking holding costs into account is less because the holding costs of the three queues are close to each other and here, Max-Weight Scheduling outperforms WFQ since the former does not let the size of any queue grow too large. Similar trends can be observed in Figures~\ref{fig:Cost_Quad_Similar} and~\ref{fig:Cost_Quad_Different}, which are for the case where the transmission costs are quadratic.

\begin{figure}[h!]
\begin{center}
\includegraphics[angle=0,scale=0.6]{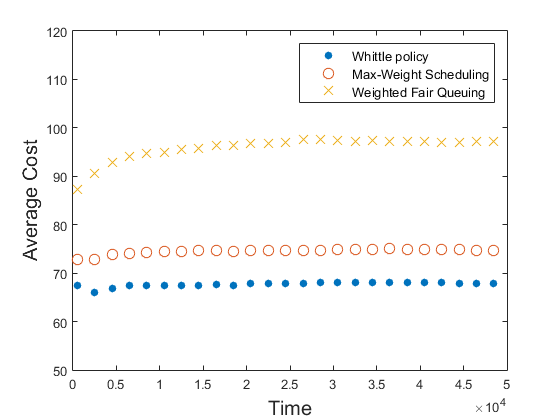}
\caption{Average cost comparison for costs 10, 20 and 30 with exponential transmission costs}
\label{fig:Cost_Exp_Similar}
\end{center}
\end{figure}

\begin{figure}[h!]
\begin{center}
\includegraphics[angle=0,scale=0.6]{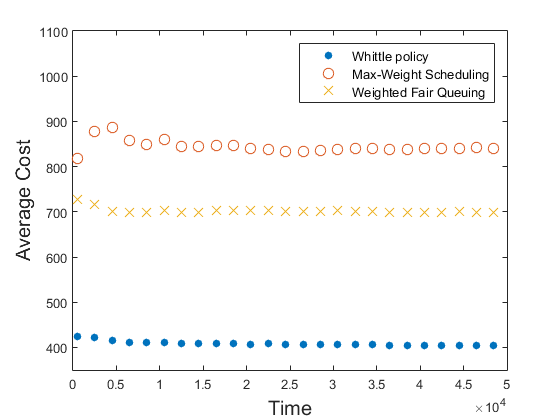}
\caption{Average cost comparison for costs 10, 20 and 500 with exponential transmission costs}
\label{fig:Cost_Exp_Different}
\end{center}
\end{figure}

\begin{figure}[h!]
\begin{center}
\includegraphics[angle=0,scale=0.6]{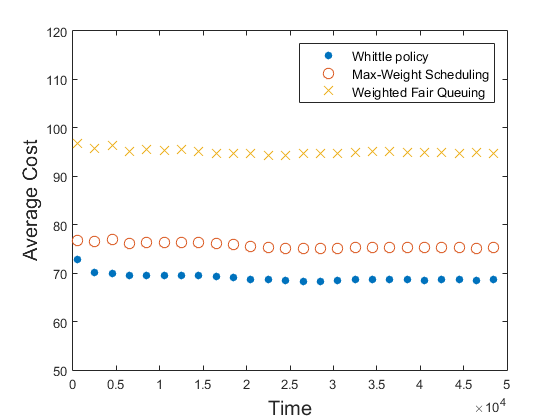}
\caption{Average cost comparison for costs 10, 20 and 30 with quadratic transmission costs}
\label{fig:Cost_Quad_Similar}
\end{center}
\end{figure}

\begin{figure}[h!]
\begin{center}
\includegraphics[angle=0,scale=0.6]{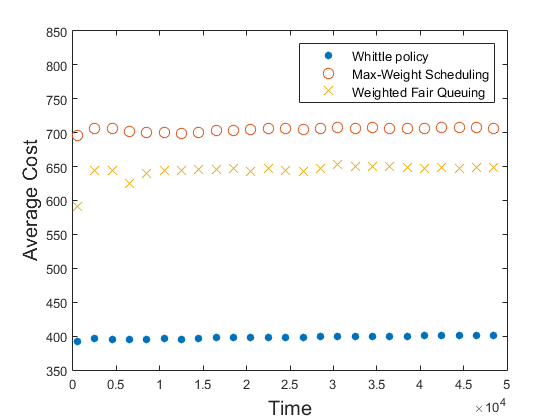}
\caption{Average cost comparison for costs 10, 20 and 500 with quadratic transmission costs}
\label{fig:Cost_Quad_Different}
\end{center}
\end{figure}

In Figure~\ref{fig:Packets_Dropped}, we have plotted the average number ( averaged over time) of packets dropped from the system (from all the three queues) against the input arrival rate for the three algorithms. It can be seen that the  Max-Weight policy drops the least number of packets, which is consistent with intuition since it selects the longest queue for transmission in each time slot, and hence keeps a check on the length of the longest queue. Also, we see that the Whittle index based algorithm performs better than WFQ in terms of the number of packets that are dropped.

\begin{figure}[h!]
\begin{center}
\includegraphics[angle=0,scale=0.6]{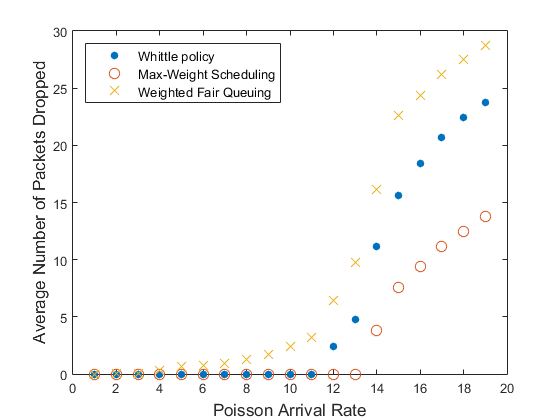}
\caption{Average number of packets dropped in the three scheduling strategies}
\label{fig:Packets_Dropped}
\end{center}
\end{figure}

\section{Conclusions}
\label{section_conclusions}

We have cast the problem of opportunistic scheduling as a restless bandit problem in the classic framework laid down by Whittle, with an additional twist that it combines another ongoing optimization, that over number of packets transmitted, over and above the bandit selection. Thus it is a `controlled' restless bandit problem. We prove Whittle indexability of this problem and propose a numerical scheme for computing Whittle indices. It would be good to have an explicit expression for Whittle indices, but that issue remains open for the moment. The index policy is empirically found to outperform some natural heuristics. Although the Whittle  heuristic is a major saving in complexity over the original problem formulation with a per stage constraint, the computational scheme for  obtaining Whittle indices still remains a cumbersome exercise. An important future direction is to explore the possibility of exploiting techniques from reinforcement learning for approximate dynamic programming for the purpose \cite{Chadha}.

Another interesting and important problem is a theoretical analysis of our heuristic for number of packets to be transmitted when active. While intuitively appealing, we do not have a rigorous justification for it at present.

\end{document}